\documentclass[12pt]{amsart}
\usepackage{amsmath,amsfonts,amssymb} \usepackage{color}
\usepackage[all,cmtip]{xy}
\usepackage{graphicx}
 \usepackage{youngtab}

 \newcommand{\bR}{\mathbb{R}}

 \newcommand{\bP}{\mathbb{P}} 
 \newcommand{\cP}{\mathcal{P}}

 \newcommand{\bZ}{\mathbb{Z}}
 \newcommand{\bC}{\mathbb{C}}
 \newcommand{\pd}{\partial}

\newcommand{\be}{\begin{equation}}
\newcommand{\ee}{\end{equation}}
\newcommand{\bea}{\begin{eqnarray}}
\newcommand{\eea}{\end{eqnarray}}
\newcommand{\ben}{\begin{eqnarray*}}
\newcommand{\een}{\end{eqnarray*}}

\newcommand{\half}{\frac{1}{2}}

\newtheorem{cor}{Corollary}[section]

\newtheorem{lem}[cor]{Lemma}
 
 \newtheorem{thm}[cor]{Theorem}
\theoremstyle{remark}

\definecolor{A}{rgb}{.75,1,.75}

\definecolor{green}{rgb}{0,1,0}
\definecolor{yellow}{rgb}{1,1,0}
\definecolor{orange}{rgb}{1,.7,0}
\definecolor{red}{rgb}{1,0,0}
\definecolor{white}{rgb}{1,1,1}

 \makeindex
\begin{document}
\title
{On Regularized Elliptic Genera of ALE Spaces}

\author{Jian Zhou}
\address{Department of Mathematical Sciences\\Tsinghua University\\Beijng, 100084, China}
\email{jzhou@math.tsinghua.edu.cn}

\begin{abstract}
We define regularized elliptic genera of ALE space of type A
by taking some regularized  nonequivariant limits of
their equivariant elliptic genera with respect to some torus actions.
They turn out to be multiples of the elliptic genus of a K3 surface.
\end{abstract}

\maketitle

\section{Introduction}

In an early paper \cite{Zho}
we have studied equivariant elliptic genera of toric Calabi-Yau 3-folds.
In this paper we will study the case of toric Calabi-Yau 2-folds. 
They are ALE spaces of type A.
The study of elliptic genera of noncompact Calabi-Yau manifolds 
was initiated about ten about years by Eguchi and Sugawara \cite{Egu-Sug, Egu-Sug2}.
The case of ALE spaces was further studied by Eguchi-Sugawara-Taormina \cite{Egu-Sug-Tao}.
Predictions for elliptic genera of toric Calabi-Yau 3-folds and ALE spaces 
were made in these papers.
Equivariant elliptic genera of noncompact spaces 
were studied around the same time in the physics literature \cite{Hol-Iqb-Vaf} 
and the mathematical literature \cite{Li-Liu-Zho, Wae1, Wae2}.
Prompted by some recent work on equivariant elliptic genera by physicists \cite{Hoh-Iqb, Har-Lee-Mur},
we attempt to verify the predictions of Eguchi and Sugawara
based on equivariant elliptic genera of toric Calabi-Yau manifolds in  and this paper.
It turns out to be successful in the case of toric Calabi-Yau 3-folds \cite{Zho},
but not so successful in the case of toric Calabi-Yau surfaces studied in this paper.
Nevertheless we believe this study has some independent interest
which might be useful for other problems.

A toric Calabi-Yau n-fold is a noncompact complex manifold 
constructed by adding some lower dimensional complex tori to the 
$n$-dimensional complex torus $(\bC^*)^n$ according to some combinatorial data.
It admits natural $n$-torus action with isolated fixed points.
One can define its equivariant elliptic genus by adding up the contributions from 
the fixed points.
The equivariant elliptic genus then depends on $n$ extra variables $t_1, \dots, t_n$,
which are linear coordinates on the Lie algebra of the $n$-torus.
Our idea is to get rid of these extra variables in some suitable way.
First of all,
we restrict to an $n-1$-dimensional subtorus which fixes a holomorphic volume form.
This not only reduces the number of extra variables by one,
but also simplifies the expressions to work with.  

In the case of toric Calabi-Yau 3-fold,
we have used a simple averaging method \cite{Zho}.
Unfortunately,
this method no long works for very simple reasons
in the case of toric Calabi-Yau surfaces or four-olds.
We will treat in this paper the case of toric Calabi-Yau surfaces by defining 
a regularized nonequivariant limit.
More precisely,
we will analyze the singularities  of the equivariant elliptic genera 
of ALE spaces of type A when the extra variable is approaching zero, 
we subtract the singular terms then take the limit. 
(In doing so we make use of some results for the equivariant elliptic genus of $\bC^2$. 
We present some alternative proofs of them to gain more understandings
from different perspectives.)
The results turn out to be proportional to the elliptic genus 
of a K3 surface. 
So they are not the same as the elliptic genera of ALE spaces 
corresponding to the decompactified elliptic genus of K3 surfaces 
studied in \cite{Egu-Sug-Tao}.
Nevertheless they should also be interesting because they are weak Jacobi forms
as in the compact case.

We arrange the rest of the paper  as follows.
In \S 2 we recall the 2-torus actions on ALE spaces of type A.
In \S 3 we present the modular transformation properties 
of their equivariant elliptic genera. 
In \S 4 we introduce their regularized elliptic genera
after analyzing the singularities of their equivariant elliptic genera.

\vspace{.1in}

{\em Acknoledgements}.
This research is partially supported by NSFC grant 11171174.
The author thanks Professor Eguchi for explaining  his joint work
with Professor Sugawara about ten years ago.

\section{A Torus Action on ALE Spaces of Type A}

In this section we  recall the definition of some  torus actions on ALE spaces of type $A$.
See \cite[Example 4.1]{Wan-Zho} and the references therein.

\subsection{The ALE Space of type $A_{n-1}$}

Such a space can be obtained by gluing $n$ copies of $\bC^2$ as follows.
For $i=0, 1, \dots, n-1$, denote by $U_i$ the $i$-th copy of $\bC^2$
and let $(x_i, y_i)$ be the linear coordinates on it.
Then $U_i$ and $U_{i+1}$ are glued together
by the following formula for change of coordinates:
\begin{align} \label{eqn:Coordinate change}
x_{i+1} & = x_i^2 y_i, & y_{i+1} & = x_i^{-1},
\end{align}
identifying $U_i - \{x_i=0\} $ with $U_{i+1} - \{y_{i+1} = 0\}$.
It is easy to see that
\be
dx_{i+1} \wedge d y_{i+1} = - dx_i \wedge dy_i,
\ee
so one has a holomorphic volume form $\Omega$ on $X_n = U_0 \cup \cdots \cup U_{n-1}$
such that on each $U_i$,
\be \label{eqn:Volume}
\Omega = (-1)^{i-1} dx_i \wedge d y_i.
\ee

For $i=1, \dots, n$,
let $C_i$ be $\{(x_{i-1}, 0) \in U_{i-1} \;|\; x_{i-1} \neq 0\}$
glued with $\{(0, y_i)\in U_i \;|\; y_i \neq 0\}$.
Then $C_i$ is a copy of $\bP^1$.
By  \eqref{eqn:Coordinate change}
one sees that the normal bundle of $C_i$ in $X_n$ is $-2$,
and so the self-intersection number of $C_i$ is $-2$.
One can also see that
\be
C_i\cdot C_j = \begin{cases}
-2, & \text{if $i = j$}, \\
1, & \text{if $i=j+1$ or $j=i+1$}, \\
0, & \text{otherwise}.
\end{cases}
\ee

For $i=1, \dots, n$,
let $V_i = U_i - \{x_i = 0\} - \{y_i = 0\} \cong (\bC^*)^2$.
Then the gluing identifies all $V_i$'s with each other,
and $X_n$ is obtained by adding copies of $(\bC^*)^1$ and $(\bC^*)^0$,
i.e., it is an example of toric varieties.

Next we recall the natural map $\pi: X_n \to \bC^2/\bZ_n$,
where $\bZ_n$ acts on $\bC^2$ as follows:
\be
\xi_n \cdot (x, y) = (\xi_n x, \xi_n^{-1} y),
\ee
where $\xi_n = e^{2\pi i/n}$.
In $U_i$, $\pi$ is given by:
\begin{align}
x & = \frac{x_iy_i}{(x_i^iy_i^{i+1})^{1/n} } = x_i^{1-i/nr} y_i^{1- (i+1)/n}, &
y & = (x_i^iy_i^{i+1})^{1/n}.
\end{align}
One can easily check that this is well-defined by \eqref{eqn:Coordinate change}.
Furthermore,
$\pi$ maps all the curves $C_i$ to $[(0,0)] \in \bC^2/\bZ_n$.
On $V_i$ we have an inverse map of $\pi$:
\begin{align} \label{eqn:xiyi}
x_i & = \frac{x^{i+1}}{y^{n-i-1}}, & y_i & = \frac{y^{n-i}}{x^i}.
\end{align}

\subsection{Torus action on ALE spaces}

Consider the following two-torus action on $\bC^2$:
\be
(s_1, s_2) \cdot (x, y) = (s_1x, s_2y).
\ee
It induces  the following $T^2$-action on $X_n$ by \eqref{eqn:xiyi}:
\be \label{eqn:Action}
(s_1, s_2) \cdot (x_i, y_i)
= (s_1^{-i-1}s_2^{n-i-1}x_i, s_1^i s_2^{-n+i}y_i)
\ee
on $U_i$.
One can directly check that these actions on $U_i$ glued together to an action on $X_n$
by \eqref{eqn:Coordinate change}.

\subsection{Weight decomposition at the fixed points}

It is clear that the $T^2$-action has $n$ isolated fixed points
$p_i$, $i=0, \dots, n-1$,
where $p_i \in U_i$ is given by $x_i =y_i = 0$.
Since $(x_i, y_i)$ are linear coordinates on $U_i$,
the weight decomposition at $p_i$ is given by \eqref{eqn:Action}.

\section{Equivariant Elliptic Genera of ALE Spaces}

In this section we study the modular transformation properties 
of equivariant genera of ALE spaces of type A by Jacobi theta function.
For reference on equivariant elliptic genera of noncompact complex manifold,
the reader can consult \cite{Li-Liu-Zho}.

\subsection{Equivariant elliptic genera of ALE spaces}

Using information on the fixed points and weight decompositions 
there,
the equivariant elliptic genus of $X_r$ is given by
the Lefschetz fixed point contributions \cite{Ati-Bot}:
\be \label{eqn:Xr-EG}
\begin{split}
& Z_{X_r}(\tau, z; t_1, t_2) \\
= & \sum_{j=0}^{r-1} \frac{\theta_1(\tau, z-(j+1)t_1+(r-j-1)t_2)}
{\theta_1(\tau, -(j+1)t_1 + (r-j-1)t_2)} \\
& \cdot  \frac{\theta_1(\tau, z+jt_1 + (j-r)t_2)}{\theta_1(\tau, jt_1 + (j-r)t_2)},
\end{split}
\ee
where $\theta_1$ is the theta function defined by:
\be
\begin{split}
\theta_1(\tau, z)
= & i q^{1/8} y^{-1/2} \prod_{m=1}^\infty (1 - q^m)(1 - y q^{m-1})(1 - y^{-1}q^m).
\end{split}
\ee
For example,
\be \label{eqn:Z-conifold}
\begin{split}
Z_{X_2}(z, \tau; t_1, t_2)
= & \frac{\theta_1(z-t_1+t_2, \tau)}{\theta_1(-t_1+t_2, \tau)}
\frac{\theta_1(z-2t_2, \tau)}{\theta_1(-2t_2, \tau)}  \\
+ & \frac{\theta_1(z-2t_1, \tau)}{\theta_1(-2t_1, \tau)} \cdot
\frac{\theta_1(z+t_1-t_2, \tau)}{\theta_1(t_1-t_2, \tau)} .
\end{split}
\ee
Such expressions were obtained in \cite{Har-Lee-Mur} in the setting of
gauged linear sigma model by the method developed in \cite{Ben-Eag-Hor-Tac}.

\subsection{Modular transformation properties of equivariant elliptic genera of
ALE spaces}

Recall that the theta-function $\theta_1$ has the following well-known
modular transformation properties \cite{Cha}:
\begin{align}
\theta_1(\tau, z+1) & = - \theta_1(\tau, z), &
\theta_1(\tau, z+\tau) & = - e^{-2\pi i z - \pi i \tau} \theta_1(\tau, z), \\
\theta_1(\tau+1, z) & = e^{\pi i/4} \theta_1(\tau, z), &
\theta_1(- \frac{1}{\tau}, \frac{z}{\tau})
& = - i \sqrt{\frac{\tau}{i}} e^{\frac{\pi i z^2}{\tau}} \theta_1(\tau, z).
\end{align}
From these it is straightforward to deduce the following transformation formulas:
\bea
&& Z_{X_n} (\tau, z+1; t_1, t_2) = (-1)^2 \cdot Z_{X_n}(\tau, z; t_1, t_2), \\
&& Z_{X_n}(\tau, z+\tau; t_1, t_2) = e^{-\pi i (t_1+t_2)} \cdot (-e^{-2\pi iz - \pi i\tau})^2 \cdot Z_{X_n}(\tau, z; t_1, t_2), \\
&& Z_{X_n}(\tau, z; t_1+1, t_2) =  Z_{X_n}(\tau, z, \tau; t_1, t_2), \\
&& Z_{X_n}(\tau, z; t_1+\tau, t_2) = e^{2\pi i z}\cdot Z_{X_n}(\tau, z; t_1, t_2), \\
&& Z_{X_n}(\tau, z; t_1, t_2+1) =  Z_{X_n}(\tau, z; t_1, t_2), \\
&& Z_{X_n}(\tau, z; t_1, t_2+\tau) =  e^{2\pi i z}\cdot Z_{X_n}(\tau, z; t_1, t_2), \\
&& Z_{X_n}(\tau+1, z; t_1, t_2) = Z_{X_n}(\tau, z; t_1, t_2), \\
&& Z_{X_n}(-\frac{1}{\tau}, \frac{z}{\tau}; \frac{t_1}{\tau}, \frac{t_2}{\tau})
= e^{2 \cdot \frac{\pi i}{\tau}(z^2-z(t_1+t_2))} Z_{X_n}(\tau, z; t_1, t_2).
\eea

\section{Regularized Elliptic Genera of $X_r$}

We define regularized nonequivariant limit of equivariant elliptic genera of 
$X_r$ by analyzing the singularity of $Z_{X_r}$.
We are based on some results of Hoheneger-Iqbal \cite{Hoh-Iqb}
for which we present some new proofs.

\subsection{Singularity of equivariant elliptic genera of $X_r$}

Using the fact that 
\be
\theta_1(\tau, -z) = -\theta_1(\tau, z),
\ee
one can rewrite \eqref{eqn:Xr-EG} in the following form:
\be
Z_{X_r}(\tau,z; t_1, t_2)
= \frac{\cdots}{\prod_{j=0}^r \theta_1(\tau, j t_1+ (j-r) t_2)}. 
\ee
We first consider $Z_{X_r}$ as a meromorphic function in $t_1$. 
Since $\theta_1(\tau, z)$ has a first order zero at $z= 2\pi i ( m \tau + n)$ for $m, n \in \bZ$,
so by \eqref{eqn:Xr-EG},
$Z_{X_r}$ has possible first order poles when
\be
jt_1 +(j-r) t_2 = m \tau + n,
\ee
i.e., 
\be
t_1 = \frac{1}{j} (m \tau +n  + (r- j) t_2).
\ee
for some $j=1, \dots, r$, $m, n \in \bZ$.
However, by calculating the residues at these points,
one finds that for many of them, 
the singularities of different terms in \eqref{eqn:Xr-EG} cancel with each other,
and only $t_1 = \frac{1}{r}(m\tau +n)$ ($m,n \in \bZ$) are the first order poles 
of $Z_{X_r}$.
The same holds for $t_2$ by symmetry between $t_1$ and $t_2$.
So one can gets an expansion of the form near $t_1=t_2 = 0$:
\be \label{eqn:Expansion}
Z_{X_r}(\tau, z; t_1, t_2)
= \frac{1}{t_1t_2} \sum_{m,n \geq 0} a_{m,n}(\tau, z) t_1^m t_2^n,
\ee
where $a_{m,n}(\tau, z)$ satisfy:
\be
a_{m,n}(\tau, z) = a_{n, m}(\tau, z).
\ee

\subsection{Equivariant elliptic genera restricted to a circle}

The singularity of $Z_{X_r}(\tau,z; t_1, t_2)$ prevents us from
taking the nonequivariant limit $\lim_{t_1=0} \lim_{t_2 \to 0}$.
To fix this problem, 
we take $t_1 = - t_2 = t$ to get:
\be
Z_{X_r}(\tau, z; t, -t)
=  r \cdot \frac{\theta_1(\tau, z+rt)}{\theta_1(\tau, rnt)}
 \cdot  \frac{\theta_1(\tau, z-rt)}{\theta_1(\tau,- rn t)},
\ee
or equivalently,
\be
\begin{split}
& Z_{X_r}(\tau, z; t, -t) \\
= & r y^{-1}  \prod_{m=1}^\infty
\frac{(1-ye^{2\pi ir t}q^{m-1})(1-y^{-1}e^{-2\pi i rt}q^{m})}
{(1- e^{2\pi irt} q^{m-1})(1- e^{-2\pi i rt}q^m)} \\
 & \cdot \prod_{m=1}^\infty \frac{
(1-ye^{-2\pi i r t}q^{m-1})(1-y^{-1} e^{2\pi i rt} q^m)}
{(1- e^{-2\pi i r t} q^{m-1})(1- e^{2\pi i rt}q^m)}.
\end{split}
\ee
This corresponds to restriction to the subtorus whose action preserves 
the holomorphic volume form \eqref{eqn:Volume}.
The expansion \eqref{eqn:Expansion} now becomes an expansion of the following form:
\be \label{eqn:Expansion2}
Z_{X_r}(\tau, z; t, -t)
= \sum_{g\geq 0} t^{2g-2} \alpha_{2g}(\tau, z).
\ee
The leading term is easily seen to be given by:
\be
\alpha_0(\tau, z) = - \frac{1}{(2\pi)^2r} 
\biggl( \frac{\theta_1(\tau, z)}{\eta(\tau)} \biggr)^2
\ee

\begin{lem} \label{lm:Ellipticity}
The coefficients $\alpha_{2g}(\tau,z)$ are weak Jacobi forms 
of index $1$ and weight $2g-2$.
Furthermore,
$b_{2g}(\tau, z) = \frac{\alpha_{2g}(\tau, z)}{\alpha_0(\tau, z)}$ 
are elliptic functions such that
\bea
&& b_{2g}(\tau+1, z) = b_{2g}(\tau, z), \\
&& b_{2g}(-\frac{1}{\tau}, \frac{z}{\tau}) = \tau^{2g} \cdot b_{2g}(\tau, z).
\eea
\end{lem}

\begin{proof}
Easy consequence of the modular transformation properties 
of $Z_{X_r}(\tau, z; t_1, t_2)$ and the expansion \eqref{eqn:Expansion2}.

\end{proof}

As noted  in \cite{Hoh-Iqb} in the case of $\bC^2$,
$Z_{X_r}(\tau, z; t, -t)$ can be expressed in terms of Eisenstein series or Weierstrass $\cP$-function:

\begin{thm}
The equivariant elliptic genus of $X_r$ can be expressed as follows:
\begin{multline} \label{eqn:Z-in-Eisenstein}
 Z_{X_r}(\tau, z; t, -t)
= \frac{-r}{(2\pi i rt)^2} \cdot y^{-1} \biggl(\prod_{n=1}^\infty \frac{(1- y q^{n-1}) (1 - y^{-1} q^n)}{(1-q^n)^2} \biggr)^2 \\
 \cdot \exp \biggl(-2 \sum_{n=1}^\infty \frac{ (rt)^{2n}}{(2n)!} \biggl( \frac{(2n-1)!}{z^{2n}} \\
-   \sum_{l=1}^\infty \frac{z^{2l}}{(2l)!} \frac{(2\pi i)^{2n+2l} B_{2n+2l}}{(2n+2l)} E_{2n+2l}(\tau) \biggr)   \biggr),
\end{multline}
where $E_{2k}(\tau)$ are the normalized Eisenstein series defined by:
\be
E_{2k}(\tau) = 1 - \frac{4k}{B_{2k}} \sum_{n=1}^\infty \sigma_{2k-1}(n) q^n,
\ee
Alternatively,
\be \label{eqn:Z-in-Weierstrass}
\begin{split}
Z_{X_r}(\tau, z; t, -t) =
& \frac{-r}{(2\pi  rt)^2}  \cdot \biggl( \frac{\theta_1(\tau, z)}{\eta(q)^3}\biggr)^2 \\
& \cdot \exp \biggl(-2 \sum_{n=1}^\infty \frac{(rt)^{2n}}{(2n)!} \frac{\pd^{2n-2}}{\pd z^{2n-2}}\cP(\tau,z)\biggr) \\
& \cdot \exp \biggl( -2 \sum_{n=2}^\infty \frac{(2\pi ir t)^{2n}}{(2n)!} \frac{B_{2n}}{2n} E_{2n}(\tau) \biggr),
\end{split}
\ee
where $\cP(\tau, z)$ is the Weierstrass $\cP$-function given by:
\be
\cP(\tau, z) = \frac{1}{z^2}
- \sum_{n=1}^\infty (2n+1) \frac{(2\pi i)^{2n+2}B_{2n+2}}{(2n+2)!} E_{2n+2}(\tau) z^{2n}.
\ee
\end{thm}

\begin{proof}
In \cite{Hoh-Iqb} a formula relating the Weierstrass $\cP$-function
and the Jacobi theta function was used. 
Here we show that these formulas can be derived by elementary series manipulations.
We begin with:
\ben
\log (1-e^v q^n)
& = &  -\sum_{m=1}^\infty \frac{q^{mn} e^{mv}}{m}
= - \sum_{m=1}^\infty \frac{q^{mn}}{m}
- \sum_{m=1}^\infty q^{mn} \sum_{k=1}^\infty m^{k-1} \frac{v^k}{k!}  \\
& = & \log (1 - q^n) - \sum_{k=1}^\infty \frac{v^k}{k!} \sum_{m=1}^\infty m^{k-1} q^{mn}.
\een
Note in passing that
\be
\sum_{m=1}^\infty m^n x^n = \frac{\sum_{j=0}^{n-1} A(n, j) x^{j+1}}{(1-x)^{n+1}},
\ee
where $A(n, j)$ are Eulerian numbers and $\sum_{j=0}^{n-1} A(n, j) x^j$ are
the Eulerian polynomials.

\ben
&& \sum_{n=1}^\infty \biggl( \log(1- e^{v}q^n) + \log (1-  e^{-v} q^n) \biggr)  \\
& = & 2 \sum_{n=1}^\infty \log (1-q^n) - 2 \sum_{k=1}^\infty \frac{v^{2k}}{(2k)!} \sum_{m, n=1}^\infty m^{2k-1} q^{mn} \\
& = & 2 \sum_{n=1}^\infty \log (1-q^n) - 2 \sum_{k=1}^\infty \frac{v^{2k}}{(2k)!}
\sum_{n=1}^\infty \sigma_{2k-1}(n) q^{n}.
\een

Now note
\ben
\frac{d}{dv} \log \frac{1-e^v}{v}
& = & -\frac{e^v}{1-e^v} - \frac{1}{v}
= 1 + \frac{1}{v} \biggl( \frac{t}{e^v-1} - 1 \biggr) \\
& = & 1 + \sum_{k=1}^\infty \frac{B_k}{k!} v^{k-1}
= \half + \sum_{k=1}^\infty \frac{B_{2k}}{(2k)!} v^{2k-1}.
\een
And so one can get:
\be \label{eqn:Log}
\log (1-e^v)
= \log v + \pi i + \half  v + \sum_{k=1}^\infty \frac{B_{2k}}{(2k)!(2k)} v^{2k}.
\ee
It follows that
\ben
&& \log  (1- e^{v}) + \log (1-  e^{-v})
=2 \log v+  2 \sum_{k=1}^\infty \frac{B_{2k}}{(2k)!(2k)} v^{2k}.
\een

Combining the above computations:
\ben
&& \log (1 - e^v) + \log (1- e^{-v}) \\
& + & 2 \sum_{n=1}^\infty \biggl( \log(1- e^{v}q^{n}) + \log (1-  e^{-v} q^n) \biggr)  \\
& = & 2\log v + 2 \sum_{k=1}^\infty \frac{B_{2k}}{(2k)!(2k)} v^{2k} \\
& + & 4 \sum_{n=1}^\infty \log (1-q^n) - 4 \sum_{k=1}^\infty \frac{v^{2k}}{(2k)!}
\sum_{n=1}^\infty \sigma_{2k-1}(n) q^{n} \\
& = & 2 \log v + 4 \sum_{n=1}^\infty \log (1- q^n)
+  2 \sum_{n=1}^\infty \frac{B_{2n}}{(2n)!(2n)} v^{2n} E_{2n}(\tau).
\een

By \eqref{eqn:Log} we get:
\ben
\log (1-e^{u\pm v})
= \log (u \pm v)  + \pi i + \half  (u \pm v)
+ \sum_{k=1}^\infty \frac{B_{2k}}{(2k)!(2k)} (u \pm v)^{2k},
\een
and so
\ben
&& \log (1-e^{u + v}) + \log (1-e^{u- v}) \\
& = & \log (u^2 - v^2)  + 2\pi i + u
+ \sum_{k=1}^\infty \frac{B_{2k}}{(2k)!(2k)} ((u + v)^{2k}+(u-v)^{2k}).
\een
Therefore,
\ben
&& \sum_{n=1}^\infty \biggl( \log(1- e^{u+v}q^{n-1}) + \log (1-  e^{-u-v} q^n) \\
& + &  \log(1- e^{u-v}q^{n-1}) + \log (1-  e^{-u+v} q^n) \biggr)  \\
& = &  \log (u^2-v^2) + 2\pi i + u + 4 \log (1-q^n) \\
& + & \sum_{k=1}^\infty \frac{B_{2k}}{(2k)!(2k)} ((u+v)^{2k} +(u-v)^{2k}) E_{2k}(\tau) \\
& = &  2 \log u+ 2\pi i + u + 4 \log (1-q^n)
+ 2 \sum_{k=1}^\infty \frac{B_{2k}}{(2k)!(2k)} u^{2k} E_{2k}(\tau) \\
& - &  \sum_{n=1}^\infty \frac{1}{n} \frac{v^{2n}}{u^{2n}}
+ 2 \sum_{k=1}^\infty \frac{B_{2k}}{2k}  \sum_{j=1}^k \frac{v^{2j}}{(2j)!} \frac{u^{2k-2j}}{(2k-2j)!} E_{2k}(\tau) \\
& = & 2 \log \prod_{n=1}^\infty (1- e^u q^{n-1}) (1 - e^{-u} q^n) \\
& - & 2 \sum_{n=1}^\infty \frac{v^{2n}}{(2n)!} \biggl( \frac{(2n-1)!}{u^{2n}}
-   \sum_{l=0}^\infty \frac{u^{2l}}{(2l)!} \frac{B_{2n+2l}}{(2n+2l)} E_{2n+2l}(\tau) \biggr).
\een
The $l=0$ terms in the second line of the last equality are
\ben
&& 2 \sum_{n=1}^\infty \frac{v^{2n}}{(2n)!} \frac{B_{2n}}{2n} E_{2n}(\tau).
\een
After cancelling all these term and plugging in $u=2\pi i z$ and $v= 2\pi i rt$,
one can prove \eqref{eqn:Z-in-Eisenstein}.
To prove \eqref{eqn:Z-in-Weierstrass},
note
\ben
&&  \frac{(2n-1)!}{u^{2n}}
-   \sum_{l=0}^\infty \frac{u^{2l}}{(2l)!} \frac{B_{2n+2l}}{(2n+2l)} E_{2n+2l}(\tau) \\
& = & \frac{\pd^{2n-2}}{\pd z^{2n-2}} \cP(\tau, z) - \delta_{n,1} \frac{B_{2}}{2} E_{2}(\tau).
\een
\end{proof}

Let
\be
G_{2k}(\tau) = - \frac{(2\pi i)^{2k}B_{2k}}{(2k)!} E_{2k}.
\ee
Then \eqref{eqn:Z-in-Weierstrass} can be rewritten as follows:
\be \label{eqn:Z-in-Weierstrass2}
\begin{split}
Z_{X_r}(\tau, z; t, -t) =
& \frac{r}{(2\pi  rt)^2}  \cdot \biggl( \frac{\theta_1(\tau, z)}{\eta(q)^3}\biggr)^2 \\
& \cdot \exp \biggl(-2 \sum_{n=1}^\infty \frac{(rt)^{2n}}{(2n)!} \frac{\pd^{2n-2}}{\pd z^{2n-2}}\cP(\tau,z)\biggr) \\
& \cdot \exp \biggl( \sum_{n=2}^\infty \frac{(rt)^{2n}}{n} G_{2n}(\tau) \biggr).
\end{split}
\ee
To find the coefficients $\alpha_{2g}(\tau,z)$ in \eqref{eqn:Expansion2},
we use the following result  due to \cite{Hoh-Iqb}:
\be \label{eqn:HI}
\begin{split}
& \exp \biggl(-2 \sum_{n=1}^\infty \frac{t^{2n}}{(2n)!} \frac{\pd^{2n-2}}{\pd z^{2n-2}}\cP(\tau,z)\biggr) 
  \cdot \exp \biggl( \sum_{n=2}^\infty \frac{t^{2n}}{n} G_{2n}(\tau) \biggr) \\
= & 1 - \cP(\tau, z) t^2 
+ \sum_{n\geq 2} (2n-1) G_{2n}(\tau) t^{2n}.
\end{split}
\ee
We will present two new proofs here. 
Let us first rewrite this formula. Since
\be
\cP(\tau, z) = \frac{1}{z^2} + \sum_{n \geq 2} (2n-1)G_{2n}(\tau) z^{2n-2},
\ee
the above formula can be rewritten as 
\be
\begin{split}
& \cP(\tau, t) - \cP(\tau, z) \\
= & \frac{1}{t^2}  \exp \biggl(-2 \sum_{n=1}^\infty \frac{t^{2n}}{(2n)!} \frac{\pd^{2n-2}}{\pd z^{2n-2}}\cP(\tau,z)
 + \sum_{n=2}^\infty \frac{t^{2n}}{n} G_{2n}(\tau) \biggr).
\end{split}
\ee
The right-hand side of this formula can be written as:
\ben 
&& \frac{1}{t^2} \exp \biggl(-2 \sum_{n=1}^\infty \frac{t^{2n}}{(2n)!} \biggl( \frac{(2n-1)!}{z^{2n}} 
+  \sum_{l=1}^\infty \frac{(2n+2l-1)!}{(2l)!}  G_{2n+2l}(\tau) z^{2l} \biggr)   \biggr) \\
& = & (\frac{1}{t^2} - \frac{1}{z^2} ) \cdot \exp \biggl(-2 \sum_{n=1}^\infty \frac{t^{2n}}{(2n)!}   
\sum_{l=1}^\infty \frac{(2n+2l-1)!}{(2l)!}  G_{2n+2l}(\tau) z^{2l} \biggr),
\een
by explicitly expanding the second term on the right-hand side,
one can see that it can be written in the following form:
\be
\frac{1}{t^2} - \cP(\tau, z) + \sum_{n \geq 2} \beta_{2n}(\tau, z) t^{2n-2}, 
\ee
where $\beta_{2n}(\tau, z)$ is holomorphic in both $\tau$ and $z$.

\begin{lem}
For each $n \geq 2$,
$\beta_{2n}(\tau, z)$ is independent of $z$ and is a modular form of weight $2n$ for $SL_2(\bZ)$,
and so it can be written simply as $\beta_{2n}(\tau)$.
\end{lem}

\begin{proof}
Similar to Lemma \ref{lm:Ellipticity},
one can show that $\beta_{2n}(\tau, z)$ is an elliptic function in $z$:
$$ \beta_{2n}(\tau, z+1) = \beta_{2n}(\tau, z+\tau) = \beta_{2n}(\tau, z),
$$
 and 
$$\beta_{2n}(-\frac{1}{\tau}, \frac{z}{\tau}) = \tau^{2n} \beta_{2n}(\tau, z).$$
But $\beta_{2n}$ is holomorphic in $z$,
so $\beta_{2n}$ is independent of $z$,
and therefore
it is a modular form of weight $2n$.
\end{proof}

Now we have:
\be
\begin{split}
& \frac{1}{t^2} - \cP(\tau, z) + \sum_{n \geq 2} \beta_{2n}(\tau) t^{2n-2} \\
= & (\frac{1}{t^2} - \frac{1}{z^2} ) \cdot \exp \biggl(-2 \sum_{n=1}^\infty \frac{t^{2n}}{(2n)!}
\sum_{l=1}^\infty \frac{(2n+2l-1)!}{(2l)!}  G_{2n+2l}(\tau) z^{2l} \biggr).
\end{split}
\ee
Now we take $\frac{\pd}{\pd t}$ on both sides:
\ben
&& - \frac{2}{t^3}  + \sum_{n \geq 2} (2n-2) \beta_{2n}(\tau) t^{2n-3} \\
& = & - \frac{2}{t^3} \cdot \exp \biggl(-2 \sum_{n=1}^\infty \frac{t^{2n}}{(2n)!}
\sum_{l=1}^\infty \frac{(2n+2l-1)!}{(2l)!}  G_{2n+2l}(\tau) z^{2l} \biggr) \\
& + & (\frac{1}{t^2} - \frac{1}{z^2} ) \cdot \exp \biggl(-2 \sum_{n=1}^\infty \frac{t^{2n}}{(2n)!}
\sum_{l=1}^\infty \frac{(2n+2l-1)!}{(2l)!}  G_{2n+2l}(\tau) z^{2l} \biggr) \\
&& \cdot \biggl(-2 \sum_{n=1}^\infty \frac{t^{2n-1}}{(2n-1)!}
\sum_{l=1}^\infty \frac{(2n+2l-1)!}{(2l)!}  G_{2n+2l}(\tau) z^{2l} \biggr) \\
& = & (\frac{1}{t^2} - \cP(\tau,z) + \sum_{n \geq 2} \beta_{2n}(\tau) t^{2n-2} )\\
&& \cdot \biggl( -\frac{2}{t} \sum_{n=0}^\infty \frac{t^{2n}}{z^{2n}}
 -2 \sum_{n=1}^\infty \frac{t^{2n-1}}{(2n-1)!}
\sum_{l=1}^\infty \frac{(2n+2l-1)!}{(2l)!}  G_{2n+2l}(\tau) z^{2l} \biggr).
\een
Since the left-hand side is independent of the variable $z$,
so after taking the terms on  the right-hand side of the second equality
independent of $z$,
we get
\ben
&& - \frac{2}{t^3}  + \sum_{n \geq 2} (2n-2) \beta_{2n}(\tau) t^{2n-3} \\
& = & -\frac{2}{t^3} + 2 \sum_{n\geq 2} (2n-1) G_{2n}(\tau) t^{2n-3} 
+ 2 \sum_{n \geq 1} \frac{(2n+1)!}{(2n-1)!2!}G_{2n+2} t^{2n-1} \\
& - & 2 \sum_{n \geq 2} \beta_{2n}(\tau) t^{2n-3} \\
& = & -\frac{2}{t^3} + \sum_{n\geq 2} (2n -2)(2n-1) G_{2n}(\tau) t^{2n-3}.  
\een
So we get:
\be
\beta_{2n}(\tau) =  (2n-1)  G_{2n}(\tau).
\ee
This completes our first proof.

For the second proof, rewrite \eqref{eqn:HI} in the following form:
\be
\begin{split}
\cP(\tau, t) - \cP(\tau, z) 
=  (2\pi)^2 \biggl( \frac{\theta_1(\tau, z)}{\eta(q)^3} \biggr)^{-2} 
\frac{\theta_1(\tau, z+t)}{\theta_1(\tau, t)} \cdot \frac{\theta_1(\tau, z-t)}{\theta_1(\tau, -t)}.
\end{split}
\ee
This can be proved using the following formula in terms of 
Weierstrass $\sigma$-function \cite[p. 55]{Cha}:
\be
\cP(\tau, u) - \cP(\tau, v)
= - \frac{\sigma(\tau, u+v)\sigma(\tau, u-v)}{\sigma^2(\tau, u)\sigma^2(\tau, v)},
\ee
and the following formula of $\sigma$  in terms of theta function \cite[p. 60]{Cha}:
\be
\sigma(\tau, u) = \frac{\theta_1(\tau, z)}{\pd_z \theta_1(\tau, 0)} e^{\eta_1 z^2}.
\ee

As a consequence of \eqref{eqn:HI},
one gets:

\begin{thm}
The expansion of $Z_{X_r}(\tau,z; t, -t)$ is given by:
\be
\begin{split}
Z_{X_r}(\tau, z; t, -t)
& = \frac{-1}{(2\pi t)^2r}
\biggl( \frac{\theta_1(\tau, z)}{\eta(\tau)} \biggr)^2 \\
& \cdot \biggl( 1 - \cP(\tau, z) (rt)^2
+ \sum_{n\geq 2} (2n-1) G_{2n}(\tau) (rt)^{2n}\biggr).
\end{split}
\ee
\end{thm}

\subsection{Regularized elliptic genera of $X_r$}

We define the regularized elliptic genus of $X_r$ to be
\be
Z^{reg}_{X_r}(\tau, z) : = \lim_{t\to 0} 
\biggl( Z_{X_r}(\tau, z; t, -t) -\frac{-1}{(2\pi t)^2r}
\biggl( \frac{\theta_1(\tau, z)}{\eta(\tau)} \biggr)^2 \biggr).
\ee
I.e.,
\be 
Z^{reg}_{X_r}(\tau, z)
= \frac{r}{(2\pi )^2}
\biggl( \frac{\theta_1(\tau, z)}{\eta(\tau)} \biggr)^2 \cP(\tau, z) .
\ee
It is a weak Jacobi form of weight $0$ and index $2$,
so it is a multiple of $Z_{K3}(\tau, z)$,
the elliptic genus of a K3 surface \cite{Egu-Oog-Tao-Yan, Bor-Lib}.

To conclude, we remark that
it should be interesting to consider 
\be \label{eqn:Expansion3}
Z_{X_r}(\tau, z; t\sqrt{\beta}, -\frac{t}{\sqrt{\beta}})
= -\frac{1}{t^2} \sum_{m,n \geq 0} a_{m,n}(\tau, z) (-1)^n 
t^{m+n}\sqrt{\beta}^{m-n},
\ee
and resum  the right-hand side as follows:
\be
\sum_{m+n =k} a_{m,n}(\tau, z) (-1)^n\sqrt{\beta}^{m-n}
= \sum_{j=0}^k c_j(\tau, z) (\sqrt{\beta} - \frac{1}{\sqrt{\beta}})^j.
\ee
It is also interesting to consider the equivariant elliptic genera of 
toric Calabi-Yau manifolds of dimensions $>3$.
We hope to adrress such problems in the future.


\begin{thebibliography}{999}


\bibitem{Ati-Bot}
M. F. Atiyal,  R. Bott,
{\em A Lefschetz fixed point formula for elliptic complexes. II. Applications}.
Ann. of Math. (2)  88  1968 451--491


\bibitem{Ben-Eag-Hor-Tac}
F. Benini, R. Eager, K. Hori, Y. Tachikawa, 
{\em Elliptic genera of 2d N=2  gauge theories}. 
Comm. Math. Phys.  333  (2015),  no. 3, 1241--1286.


\bibitem{Bor-Lib}
L.A. Borisov, A. Libgober,
{\em Elliptic genera of toric varieties and applications to mirror symmetry}.
Invent. Math.  140  (2000),  no. 2, 453--485.

\bibitem{Cha}
K. Chandrasekharan,
Elliptic functions.
Grundlehren der Mathematischen Wissenschaften, 281. Springer-Verlag, Berlin, 1985.

\bibitem{Egu-Oog-Tao-Yan}
T. Eguchi, H. Ooguri, A. Taormina, S.-K. Yang,
{\em Superconformal algebras and string compactification
on manifolds with $SU(n)$  holonomy}.
Nuclear Phys. B  315  (1989),  no. 1, 193--221.



\bibitem{Egu-Sug}
T. Eguchi, Y. Sugawara,
{\em $SL(2;\bR)/U(1)$  supercoset and elliptic genera of non-compact Calabi-Yau manifolds}.
J. High Energy Phys.  2004,  no. 5, 014, 38 pp. (electronic).

\bibitem{Egu-Sug2}
T. Eguchi, Y. Sugawara, 
{\em Conifold type singularities, N=2  Liouville and $SL(2;\bR)/U(1)$  theories}. 
J. High Energy Phys.  2005,  no. 1, 027, 50 pp.

\bibitem{Egu-Sug-Tao}
T. Eguchi, Y. Sugawara, A. Taormina, 
{\em  Modular forms and elliptic genera for ALE spaces}.  
Exploring new structures and natural constructions in mathematical physics,  125--159, 
Adv. Stud. Pure Math., 61, Math. Soc. Japan, Tokyo, 2011.


\bibitem{Har-Lee-Mur}
J.A.Harvey, S. Lee, S. Murthy, 
{\em Elliptic genera of ALE and ALF manifolds from gauged linear sigma models}. 
J. High Energy Phys.  2015,  no. 2, 110, front matter+50 pp.

\bibitem{Hoh-Iqb}
S. Hohenegger, A. Iqbal, 
{\em M-strings, elliptic genera and N=4  string amplitudes}. 
Fortschr. Phys.  62  (2014),  no. 3, 155--206.
 
\bibitem{Hol-Iqb-Vaf}
T. Hollowood, A. Iqbal, C. Vafa, Cumrun,
{\em Matrix models, geometric engineering and elliptic genera}. 
J. High Energy Phys.  2008,  no. 3, 069, 81 pp. 

\bibitem{Li-Liu-Zho}
J. Li, K. Liu, J. Zhou,
{\em Topological string partition functions as equivariant indices}.
Asian J. Math.  10  (2006),  no. 1, 81--114.

\bibitem{Wae1}
R. Waelder, 
{\em Equivariant elliptic genera}. 
Pacific J. Math.  235  (2008),  no. 2, 345--377. 

\bibitem{Wae2}
R. Waelder, 
{\em Equivariant elliptic genera and local McKay correspondences}. 
Asian J. Math.  12  (2008),  no. 2, 251--284. 

\bibitem{Wan-Zho}
Z. Wang, J. Zhou, 
{\em Tautological sheaves on Hilbert schemes of points}. 
J. Algebraic Geom.  23  (2014),  no. 4, 669--692.

\bibitem{Zho}
J. Zhou, 
{\em  On equivariant elliptic genera of toric Calabi-Yau 3-folds}, 
arXiv:1510.08528. 

 
\end{thebibliography}
\end{document}